\documentclass{article}
\usepackage{ijcai27}
\usepackage{times}
\usepackage{soul}
\usepackage{url}
\usepackage[hidelinks]{hyperref}
\usepackage[utf8]{inputenc}
\usepackage[small]{caption}
\usepackage{graphicx}
\usepackage{amsmath}
\usepackage{amssymb}
\usepackage{amsthm}
\usepackage{booktabs}
\usepackage{array}
\usepackage{algorithm}
\usepackage{algorithmic}
\newtheorem{definition}{Definition}
\newtheorem{lemma}{Lemma}
\newtheorem{proposition}{Proposition}
\newtheorem{theorem}{Theorem}
\newtheorem{corollary}{Corollary}
\newtheorem{example}{Example}
\title{Choosing What to Observe: Higher-Order Obstructions in Horn Abduction}
\author{
    Yang Yukun
    \affiliations
    School of Management, Beijing Institute of Technology
    \emails
    yangyukun2005@163.com
}
\begin{document}
\maketitle
\begin{abstract}
An explanation may exist in every complete context yet be impossible to select from the information observed. We study this distinction in finite propositional Horn abduction with integrity constraints. Observations are sufficient when each report admits one hypothesis set that is consistent and entails the manifestation in every compatible context. We characterize sufficiency by minimal context families with no common explanation. Their variable spans form a canonical hypergraph whose transversals are exactly the sufficient observation sets. The characterization recovers classical discernibility for unique explanations and supports composition across independent components. Minimal obstructions have unbounded cardinality even for acyclic rules, while rectangular explanation families admit a tight bound given by their largest choice domain. With circuit-defined contexts, sufficiency is $\Pi_3^P$-complete and selecting at most $k$ observable variables is $\Sigma_4^P$-complete, already for unrestricted contexts and acyclic rules. The selection reduction makes observation choices enforce globally fixed existential commitments. Explicit contexts reduce both tasks to NP-complete problems; tabulated explanation families make sufficiency polynomial. These results connect explanatory choice, information requirements, and representation cost.

\end{abstract}
\section{Introduction}

Abduction supplies hypotheses that explain a manifestation consistently with background knowledge \cite{poole_logical_1988,poole_explanation_1989,kakas_abductive_1992}. When background facts depend on partially observable conditions, an explanation valid in one context may fail in another producing the same report. The reasoner must therefore choose hypotheses using only the information it observes.

Consider three operating contexts $00,01,10$ with mutually exclusive explanatory choices $\{a,b\}$, $\{b,c\}$, and $\{a,c\}$, respectively. Every pair shares a choice, but no choice works in all three. Reading either telemetry bit permits a valid commitment without identifying the complete context. Thus contextwise existence and pairwise compatibility both fall short of uniform explanatory choice, even when Horn forward chaining verifies each individual choice efficiently.

We ask which source variables must be observed so that every report permits a single hypothesis set valid in all compatible contexts. We characterize failure by minimal context collections with empty explanation intersection. Their variable spans yield a canonical observation contract: sufficient variable sets are exactly its hypergraph transversals. Obstructions can be arbitrarily large under acyclic Horn rules. Independent finite choice slots instead give a tight domain-size bound, recovering pairwise discernibility for Boolean slots. Contracts also compose by conjunction across independent abductive components.

The requirement fixes the hypotheses before hidden facts are known. Entailment and consistency are evaluated separately in each remaining context; the contexts need not hold simultaneously. A contextwise choice could exploit information unavailable in the report. Pairwise tests introduce a different error: they assume an intersection property that arbitrary explanation families lack. Our characterization identifies exactly which collections must be split to avoid both errors.

With circuit-defined contexts, robust verification, uniform explanation existence for one cell, and sufficiency are respectively coNP-, $\Sigma_2^P$-, and $\Pi_3^P$-complete. Budgeted observation selection is $\Sigma_4^P$-complete. Its lower bound uses paired telemetry variables to force one globally fixed assignment while permitting other hypotheses to depend on reports. Listed contexts reduce sufficiency and selection to NP-complete problems; tabulated explanation families make sufficiency polynomial but leave selection NP-complete. These distinctions identify the cost of information dependencies beyond local Horn inference.

\section{Background and Related Work}

Classical abduction distinguishes deriving a conclusion from supplying consistent assumptions that support it \cite{console_relationship_1991,kakas_abductive_1992}. Diagnosis uses conflicts and hitting sets \cite{reiter_theory_1987}; truth maintenance records supporting environments \cite{de_kleer_assumption-based_1986}, with diagnoses depending on component assumptions \cite{de_kleer_characterizing_1992}. Our obstructions consist of contexts, and their transversals select observations. Eiter and Gottlob~\shortcite{eiter_complexity_1995} classify explanation existence, relevance, and necessity; Eiter and Makino~\shortcite{eiter_computing_2007} study enumeration from Horn theories. Tractable fragments depend on language restrictions \cite{zanuttini_new_2003,creignou_complete_2006,nordh_what_2008} and semantics \cite{eiter_abduction_1997}. Here local verification is polynomial; complexity arises from uniform choice across contexts. Facets describe variability within explanation families \cite{schmidt_complexity_2025}, whereas we compare their intersections across contexts. Classifier explanations identify sufficient feature assignments \cite{ignatiev_abduction-based_2019}; our observations instead supply information for selecting hypotheses.

Learning abductive explanations from partial examples gives statistical guarantees \cite{juba_learning_2018}; we require validity in every supplied possible context. Observational cells follow standard epistemic semantics \cite{fagin_reasoning_2003}. Uniform decisions also occur in incomplete-information games \cite{reif_complexity_1984} and partially observable planning \cite{rintanen_complexity_2004}; our setting isolates one commitment after a fixed observation. Dependency-quantified formulas formalize restricted existential information \cite{balabanov_henkin_2014}. All hypothesis bits here share one observation set, yielding an ordinary three-block formula for sufficiency. Restricted-vocabulary conditions \cite{lin_strongest_2001} and forgetting \cite{lang_propositional_2003} offer related perspectives: validating a fixed hypothesis set requires universal quantification over hidden variables, followed by a choice of hypotheses for each report.

Rough sets partition objects by observed attributes \cite{pawlak_rough_1982}, and discernibility functions encode pairwise decision differences \cite{slowinski_discernibility_1992}. We recover this construction for unique explanations. Superset decision tables allow candidate labels and study instantiations and entropy-based reducts \cite{lesot_feature_2020}. Our tabulated feasibility condition is an instantiation constant within each observation cell (Proposition 3); higher-order obstructions and succinct Horn representations supply the additional questions studied here. Diagnosability \cite{sampath_diagnosability_1995} and sensor minimization \cite{cassez_sensor_2007} study observation requirements for faulty versus nonfaulty behavior. Multiple permissible explanations instead require compatibility beyond pairs.

The decision-table connection concerns feasibility of a label selection. It does not identify entropy-based reduct criteria with budgeted observation selection. With multiple candidate explanations, preserving an original decision and selecting a common permissible decision are different objectives. We use the latter objective to expose minimal empty-intersection witnesses and then study how a Horn theory represents the candidate families. This separates the combinatorial correspondence from the cost of constructing its input.

Hypergraph transversals \cite{berge_hypergraphs_1989} and monotone dualization \cite{fredman_complexity_1996} describe the resulting contracts. As in knowledge compilation \cite{darwiche_knowledge_2002}, an exact representation may be expensive to construct. The lower bounds use standard quantified problems in the polynomial hierarchy \cite{stockmeyer_polynomial-time_1976,wrathall_complete_1976,chandra_alternation_1981}.

\section{Observation-Aware Horn Abduction}

\subsection{Contexts and classical validity}

All vocabularies are finite and partitioned into source, hypothesis, and derived atoms. Each source variable $x\in X$ has distinct atoms $x^0,x^1\in I_X$; a context $c:X\rightarrow\{0,1\}$ supplies $F_c=\{x^{c(x)}:x\in X\}$. Superscripts name atoms. Hypotheses form $H$. Rules $v\leftarrow u_1,\ldots,u_r$ have derived heads or the distinct symbol $\bot$, and bodies over the three vocabularies; empty bodies are allowed. Heads $\bot$ are integrity constraints. Source and hypothesis atoms never occur as heads, and $\bot$ never occurs in bodies.

The least Horn closure $\operatorname{cl}_T(A)$ is polynomial-time computable \cite{dowling_linear-time_1984}. The manifestation $g$ is derived; conjunctive manifestations use a fresh atom requiring all conjuncts. Acyclicity refers to the dependency graph with edges from body atoms to heads.

\begin{definition}[Contextual abduction instance]
An instance is $\mathcal P=(X,H,T,g,C)$, where $C\subseteq\{0,1\}^X$ is the set of possible contexts. A hypothesis set $S\subseteq H$ is valid at $c\in C$ if
\end{definition}

\begin{equation}
V(c,S)\quad\Longleftrightarrow\quad
g\in\operatorname{cl}_T(F_c\cup S)
\ \text{and}\ 
\bot\notin\operatorname{cl}_T(F_c\cup S).
\label{eq:validity}
\end{equation}

Let $\mathcal E_c=\{S\subseteq H:V(c,S)\}$ and $\mathcal U=2^H$, with empty-index intersection $\mathcal U$. Unless stated otherwise, $C\ne\varnothing$. Contexts are listed or circuit-defined: $C=\{c:\Gamma(c)=1\}$. Explanations may be nonminimal; $C$ records source correlations.

\subsection{Uniformity and sufficiency}

An observation scheme $O\subseteq X$ induces $c\equiv_O d$ iff $c|_O=d|_O$, with partition $C/O$. Only realized reports require a response.

\begin{definition}[Uniform explanation policy]
A policy for $O$ is a function $\pi:\{c|_O:c\in C\}\rightarrow\mathcal U$. It is valid if $V(c,\pi(c|_O))$ for every $c\in C$. The scheme $O$ is sufficient if such a policy exists. For a nonempty cell $B\in C/O$, define its uniform explanation family by
\end{definition}

\begin{equation}
\mathcal E(B)=\bigcap_{c\in B}\mathcal E_c.
\label{eq:uniform-explanation-family}
\end{equation}

\begin{lemma}[Cell criterion]
A scheme $O$ is sufficient exactly when $\mathcal E(B)\ne\varnothing$ for every $B\in C/O$.
\end{lemma}

\begin{proof}
A policy chooses one member of each cell intersection; conversely, those choices define a valid policy on the finitely many reports.
\end{proof}

\begin{proposition}[Refinement and boundaries]
Sufficiency is upward closed: if $O$ is sufficient and $O\subseteq O'$, then $O'$ is sufficient. The empty scheme is sufficient exactly when $\bigcap_{c\in C}\mathcal E_c\ne\varnothing$. Full observation $X$ is sufficient exactly when every $\mathcal E_c$ is nonempty. If $C=\varnothing$, every scheme is sufficient.
\end{proposition}

\begin{proof}
Refinement shrinks cells and preserves their common explanations. Empty and full observation give, respectively, one cell and singleton cells; an empty context space has no reports.
\end{proof}

\begin{example}[Incident triage]
Let $X=\{p,q\}$, $H=\{a,b,c\}$, and $C=\{00,01,10\}$. Include $\bot\leftarrow u,v$ for every distinct pair $u,v\in H$. Include the following six explanatory rules:
\end{example}

\begin{equation}
\begin{array}{ll}
g\leftarrow a,p^0,q^0,&g\leftarrow b,p^0,q^0,\\
g\leftarrow b,p^0,q^1,&g\leftarrow c,p^0,q^1,\\
g\leftarrow a,p^1,q^0,&g\leftarrow c,p^1,q^0.
\end{array}
\label{eq:triage-rules}
\end{equation}

This acyclic theory has families $\mathcal E_{00}=\{\{a\},\{b\}\}$, $\mathcal E_{01}=\{\{b\},\{c\}\}$, and $\mathcal E_{10}=\{\{a\},\{c\}\}$. Every pair intersects; the triple does not. Observing $p$ permits $\pi(0)=\{b\},\pi(1)=\{a\}$; observing $q$ permits $\pi(0)=\{a\},\pi(1)=\{b\}$. Sufficient schemes therefore need not be intersection-closed or have a least member.

\begin{example}[Local minimality loses uniform commitments]
Let $X=\{x\}$, $H=\{a,b\}$, $C=\{0,1\}$, and $T=\{g\leftarrow a,x^0;\ g\leftarrow b,x^1\}$, with no constraints. The locally minimal explanation families are $\{\{a\}\}$ and $\{\{b\}\}$, but $\{a,b\}$ is valid in both contexts. No observation is needed under Definition 2. Intersecting only locally minimal explanations would incorrectly require observing $x$.
\end{example}

Preferences can select an inclusion-minimal member after cellwise intersection. Such a member need not be locally minimal.

Uniformity fixes hypotheses, while derivations may vary: $g\leftarrow a,x^0$ and $g\leftarrow a,x^1$ validate $\{a\}$ in both contexts, although neither source-free derivation suffices. Requiring a proof from $T\cup\{a\}$ alone would reject this valid commitment. Conversely, accepting different hypothesis sets whenever each context has a proof would accept the empty observation scheme in Example 1. The common hypothesis set is the decision constrained by available information; its deductive support may use the actual context. Preferences therefore apply to the resulting uniform family, after the information restriction has been enforced.

\section{Explanatory Obstructions and Observation Contracts}

\subsection{From failed commitments to variable requirements}

\begin{definition}[Explanatory obstruction]
A nonempty set $D\subseteq C$ is an obstruction if $\bigcap_{c\in D}\mathcal E_c=\varnothing$. It is minimal if every nonempty proper subset has a nonempty intersection. Let $\mathcal B(\mathcal P)$ contain the minimal obstructions. For $D\subseteq C$, define
\end{definition}

\begin{equation}
\operatorname{span}(D)=\{x\in X:\exists c,d\in D,\ c(x)\ne d(x)\}.
\label{eq:obstruction-span}
\end{equation}

The canonical observation contract is the hypergraph

\begin{equation}
\mathcal H(\mathcal P)=\min_{\subseteq}
\{\operatorname{span}(D):D\in\mathcal B(\mathcal P)\},
\label{eq:observation-contract}
\end{equation}

where $\min_{\subseteq}$ removes duplicate edges and proper supersets. An empty hypergraph has no edges; a hypergraph containing the empty edge has no transversal. A transversal is a variable set meeting every edge.

\begin{theorem}[Exact observation characterization]
For every contextual abduction instance and every $O\subseteq X$,
\end{theorem}

\begin{equation}
O\text{ is sufficient}\quad\Longleftrightarrow\quad
\forall e\in\mathcal H(\mathcal P),\ O\cap e\ne\varnothing.
\label{eq:exact-observation}
\end{equation}

\begin{proof}[Proof sketch]
Suppose $O$ is sufficient. If it misses the span of a minimal obstruction $D$, all contexts of $D$ agree on $O$ and lie in one cell $B$. Any explanation common to $B$ would belong to every $\mathcal E_c$ for $c\in D$, contradicting obstructionhood. Thus $O$ meets every obstruction span and every retained edge.

Conversely, suppose $O$ meets all retained edges but some cell $B$ has empty intersection. Finiteness supplies an inclusion-minimal nonempty $D\subseteq B$ with empty intersection. It is an obstruction, and its contexts agree on $O$, so $O$ misses its span. Every obstruction span contains an inclusion-minimal span retained in Equation~\eqref{eq:observation-contract}; hence $O$ misses a contract edge, a contradiction. This also covers singleton obstructions and their empty spans.
\end{proof}

Each obstruction must be split by an observed coordinate; its contexts need not be fully distinguished.

\begin{figure*}[t]
\centering
\includegraphics[width=\textwidth]{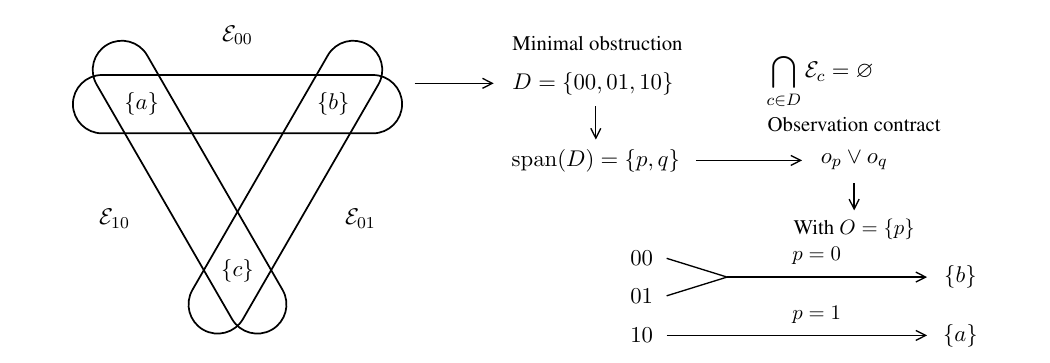}
\caption{A higher-order explanatory obstruction and its observation contract (Example~1). The three context-specific explanation families intersect pairwise but have no common member. The obstruction's span $\{p,q\}$ yields the contract $o_p\lor o_q$. The paths show a valid policy when $p$ is observed.}
\label{fig:triage-observation}
\end{figure*}

\begin{corollary}[Optimization and mandatory observations]
Inclusion-minimal sufficient schemes are exactly the minimal transversals of $\mathcal H(\mathcal P)$. For nonnegative variable costs, minimum-cost sufficient schemes are minimum-cost transversals. If all schemes are restricted to $K\subseteq X$, one exists exactly when $K$ meets every contract edge. When some sufficient scheme exists, a variable is present in every sufficient scheme exactly when its singleton is a contract edge.
\end{corollary}

\begin{proof}
The first three claims follow from Equation~\eqref{eq:exact-observation}. A singleton edge forces its variable; otherwise every edge contains another variable, so $X\setminus\{x\}$ is sufficient.
\end{proof}

In Example 1 (Figure~\ref{fig:triage-observation}), the only minimal obstruction is $\{00,01,10\}$, giving contract $o_p\lor o_q$, where $o_x$ means observing $x$. An individually infeasible context yields an empty span, making every scheme insufficient.

\subsection{Canonicality and classical correspondence}

\begin{proposition}[Observation equivalence]
For instances $\mathcal P,\mathcal P'$ on the same source variables, the same observation schemes are sufficient exactly when $\mathcal H(\mathcal P)=\mathcal H(\mathcal P')$.
\end{proposition}

\begin{proof}[Proof sketch]
By Theorem 1, equal contracts give equal sufficient schemes. Conversely, complements of irredundant edges are exactly the maximal insufficient schemes: enlarging such a complement would require an edge properly contained in the original edge. Taking complements recovers the contract, including the empty and empty-edge cases.

More explicitly, $X\setminus e$ misses edge $e$. Any strict superset missing an edge $f$ would imply $f\subsetneq e$, contrary to irredundancy. Every insufficient scheme misses some edge and is contained in its complement, so no other maximal insufficient schemes exist. An edgeless contract has no insufficient schemes; the contract $\{\varnothing\}$ has unique maximal insufficient scheme $X$.
\end{proof}

Equal contracts preserve observation feasibility; policies may still use different hypotheses.

\begin{proposition}[Decision-table correspondence]
If every $\mathcal E_c$ is a singleton, minimal obstructions are exactly pairs with different explanations, and Equation~\eqref{eq:observation-contract} is the ordinary discernibility hypergraph. More generally, in the tabulated setting, $O$ is sufficient exactly when the candidate-label table admits an instantiation constant on each $O$-cell.
\end{proposition}

\begin{proof}
Singleton families fail to intersect exactly when two labels differ. In general, a cellwise constant instantiation is precisely the choice in Lemma 1.
\end{proof}

\subsection{Obstruction size and expressive limits}

\begin{lemma}[Finite witness bound]
If $D$ is a minimal obstruction with $|D|\ge2$, there are distinct hypothesis sets $S_c\in\mathcal U$, one for each $c\in D$, such that $S_c\in\mathcal E_d$ for every $d\in D\setminus\{c\}$ but $S_c\notin\mathcal E_c$. Consequently,
\end{lemma}

\begin{equation}
|D|\le |\mathcal U|,
\qquad
|D|\le 1+\min_{c\in D}|\mathcal E_c|.
\label{eq:finite-witness-bound}
\end{equation}

\begin{proof}
Choose $S_c\in\bigcap_{d\in D\setminus\{c\}}\mathcal E_d$ by minimality. It is excluded by $\mathcal E_c$. Two such witnesses cannot coincide, and all except $S_c$ belong to $\mathcal E_c$, proving both bounds.

Indeed, if $S_c=S_d$ for distinct $c,d$, the two witness conditions place that same set in every family indexed by $D$, contradicting obstructionhood. Fixing $c$ then gives $|D|-1$ distinct members of $\mathcal E_c$.
\end{proof}

\begin{proposition}[Unbounded higher-order obstructions]
For every integer $r\ge2$, an acyclic Horn instance has $r$ contexts whose full collection is a minimal obstruction and whose proper collections all admit a common explanation.
\end{proposition}

\begin{proof}[Proof sketch]
Use hypotheses $h_1,\ldots,h_r$, pairwise exclusivity constraints, and distinct context codes $c_1,\ldots,c_r$. Add $g\leftarrow h_j,F_{c_i}$ for $i\ne j$, where $F_{c_i}$ denotes its conjunction. Then $\mathcal E_{c_i}=\{\{h_j\}:j\ne i\}$. Every proper collection shares an omitted context's hypothesis; the full collection shares none. Heads are only $g,\bot$, ensuring acyclicity.
\end{proof}

The construction attains the second bound in Equation~\eqref{eq:finite-witness-bound}. Fresh conjunction chains also permit binary bodies, so neither acyclicity nor bounded rule width bounds obstruction size.

\begin{theorem}[Realization of observation contracts]
Every hypergraph $\mathcal G$ on $X$ with no redundant edges is the observation contract of an acyclic Horn instance. If $\mathcal G$ has no empty edge, two hypothesis atoms suffice, and at most $1+|\mathcal G|$ contexts are required.
\end{theorem}

\begin{proof}[Proof sketch]
For nonempty edges use their characteristic contexts $c_e$ and the zero context $c_0$. With hypotheses $a,b$, constraint $\bot\leftarrow a,b$, and exact-guard rules $g\leftarrow a,F_{c_0}$ and $g\leftarrow b,F_{c_e}$, the minimal obstructions are exactly $\{c_0,c_e\}$, with spans $e$. No edges require only $c_0$; the empty-edge contract uses one context with no explanatory rule. These cases include $X=\varnothing$. The construction is acyclic and has $O((|\mathcal G|+1)(|X|+1))$ encoding size.
\end{proof}

Thus every upward-closed observation requirement is realizable. Obstruction cardinality and span size are distinct parameters.

\subsection{Context restriction and obstruction provenance}

Let $\mathcal P|_{C'}$ replace $C$ by $C'\subseteq C$, leaving all other data fixed.

\begin{proposition}[Context restriction]
Minimal obstructions under restricted context information satisfy
\end{proposition}

\begin{equation}
\mathcal B(\mathcal P|_{C'})
=\{D\in\mathcal B(\mathcal P):D\subseteq C'\}.
\label{eq:context-restriction}
\end{equation}

Every scheme sufficient before restriction remains sufficient afterwards. However, the canonical contract alone does not determine the contract after an arbitrary specified restriction.

\begin{proof}[Proof sketch]
Obstructionhood and minimality depend only on member contexts, proving Equation~\eqref{eq:context-restriction}; restriction shrinks cells, preserving sufficiency. For nonrecoverability, take the full Boolean square on $p,q$ with unique singleton explanations determined respectively by $p\wedge q$ and $p\vee q$. Exact guards realize both acyclically. Both contracts are $\{\{p\},\{q\}\}$. On $C'=\{00,01\}$, the first becomes empty and the second becomes $\{\{q\}\}$.
\end{proof}

Updating under context restriction requires obstruction witnesses: larger spans discarded from the canonical contract may become relevant when smaller witnesses disappear.

\section{Computational Complexity and Representation}

\subsection{Decision problems and a Horn compilation lemma}

We use polynomial-time many-one reductions. Succinct inputs contain $X,H,T,g,\Gamma$. \textbf{Sufficiency} tests a supplied $O$; \textbf{Observation Selection} asks whether a sufficient $O$ has $|O|\le k$. \textbf{Cell Existence} tests whether $\mathcal E([c]_O)\ne\varnothing$ for supplied $c\in C$ and $O$; \textbf{Robust Verification} tests membership of a supplied $S$ in that intersection. Malformed contexts are rejected. Tabulated families are semantic inputs, without a separate Horn-realization check.

\begin{lemma}[Circuit realization]
For every Boolean circuit $R(X,Y)$ with disjoint input blocks $X,Y$, a polynomial-size acyclic Horn theory with rule bodies of size at most two can be constructed with exactly the hypothesis atoms $y_i^0,y_i^1$ and one manifestation $g$. At context $c$, its valid explanations are exactly the sets $S_s=\{y_i^{s(y_i)}:y_i\in Y\}$ for complete assignments $s$ satisfying $R(c,s)$.
\end{lemma}

\begin{proof}[Proof sketch]
Use a dual-rail simulation of binary AND/OR, unary NOT, and constant gates. For $v=u\wedge w$, add $v^1\leftarrow u^1,w^1$, $v^0\leftarrow u^0$, and $v^0\leftarrow w^0$. For $v=u\lor w$, add $v^1\leftarrow u^1$, $v^1\leftarrow w^1$, and $v^0\leftarrow u^0,w^0$. NOT swaps rails; a constant supplies its value rail as a fact. Add $\bot\leftarrow y_i^0,y_i^1$ and selectors $s_i\leftarrow y_i^0$, $s_i\leftarrow y_i^1$. Derive $g$ from the positive output rail and all selectors using a binary conjunction chain; with $Y=\varnothing$, use just the output rail.

The selectors force at least one rail of every hypothesis bit, and the constraints permit at most one. Source facts and the chosen hypothesis assignment therefore supply exactly one rail per input. Induction over gates derives precisely each gate's correct rail, so $g$ follows exactly for accepting assignments. Conversely, each such assignment enables every selector and triggers no constraint. Auxiliary atoms are derived, never hypotheses. Circuit order followed by selectors, the chain, $g$, and sink $\bot$ gives acyclicity. Size is linear in the normalized circuit, with all bodies of size at most two.
\end{proof}

\subsection{Fixed observations}

\begin{theorem}[Succinct verification and sufficiency]
Robust Verification is coNP-complete, Cell Existence is $\Sigma_2^P$-complete, and Sufficiency is $\Pi_3^P$-complete. Restricting Sufficiency to $O=\varnothing$ yields $\Sigma_2^P$-completeness; restricting it to $O=X$ yields $\Pi_2^P$-completeness. Every lower bound holds for $\Gamma=\top$, acyclic Horn rules, and one manifestation.
\end{theorem}

\begin{proof}[Proof sketch]
Robust Verification universally checks compatible contexts, giving coNP membership; realizing a hypothesis-free circuit and observing nothing reduces tautology. Cell Existence has prefix $\exists S\forall d$; Lemma 3 realizes $\exists Y\forall W\,R(Y,W)$ with no observation and representative $0^W$. Empty-observation Sufficiency has the same prefix and reduction.

For observed $Z$ and hidden $W$, Lemma 1 gives $\forall z\exists S\forall w\,(\neg\Gamma(z,w)\lor V((z,w),S))$. Realizing $\forall Z\exists Y\forall W\,R(Z,Y,W)$ proves $\Pi_3^P$-hardness. Full observation instead gives $\forall c\exists S\,(\neg\Gamma(c)\lor V(c,S))$, with hardness from $\forall X\exists Y\,R(X,Y)$. All reductions use $\Gamma=\top$.

For the upper bounds, the matrix is polynomial-time decidable by evaluating $\Gamma$ and computing Horn closure. Reports with no compatible context make every implication true and require no policy response. For the lower bounds, Lemma 3 forces explanations to encode complete existential assignments, while source facts encode universal assignments. Observing exactly $Z$ gives the explanatory assignment access to $z$ but fixes it across all values of $W$, matching the quantified formula.
\end{proof}

The additional quantifier block expresses dependence on observed variables and independence from hidden ones; refinement still preserves feasibility (Proposition 1).

\subsection{Choosing observations}

For selection hardness, the chosen coordinates must encode one commitment shared by all reports.

\begin{theorem}[Succinct observation selection]
Observation Selection is $\Sigma_4^P$-complete, even when $\Gamma=\top$, $T$ is acyclic, and all rule bodies have size at most two.
\end{theorem}

\begin{proof}[Proof sketch]
Selection has prefix $\exists o\forall c\exists S\forall d$ and polynomial matrix $|o|\le k\land((\Gamma(c)\land\Gamma(d)\land c\equiv_o d)\rightarrow V(d,S))$.

Reduce $\exists A\forall B\exists Y\forall D\,\varphi(A,B,Y,D)$. For each $a_i$, introduce source pair $p_i,q_i$. All $B,D,P,Q$ vary freely. Hypothesis assignments encode $A,Y$, bits $t_i$, and $u_j$ for each $b_j$, with validity predicate:
\begin{equation}
\begin{aligned}
R={}&\varphi(A,B,Y,D)\\
&{}\land\bigwedge_i[(\neg a_i\rightarrow t_i=p_i)\land(a_i\rightarrow t_i=q_i)]\\
&{}\land\bigwedge_j(u_j=b_j).
\end{aligned}
\label{eq:selection-reduction}
\end{equation}
Apply Lemma 3 with $k=|A|+|B|$. Sufficiency forces observing each $b_j$: otherwise no fixed $u_j$ tracks it within a cell. It also forces observing at least one of each $p_i,q_i$, since either value of $a_i$ would otherwise leave $t_i$ tracking a hidden bit. The budget permits exactly these choices. Observing only $p_i$ forces $a_i=0$ in every cell; observing only $q_i$ forces $a_i=1$. Thus $O$ fixes one global assignment $A^*$. Fix any report on the chosen pair coordinates. For every $B$, its cell explanation supplies $Y$ satisfying $\varphi(A^*,B,Y,D)$ for all freely varying $D$. Conversely, a witnessing $A^*$ selects the corresponding coordinates; for each report choose the witnessing $Y(B)$, copy observed values into $t_i,u_j$, and obtain a common explanation. This proves equivalence and the claimed restrictions.
\end{proof}

\subsection{Explicit representations and policy size}

\begin{theorem}[Representation boundaries]
If $C$ is explicitly listed but explanation families remain Horn-defined, Sufficiency and Observation Selection are NP-complete. If every $\mathcal E_c$ is explicitly tabulated as well, Sufficiency is polynomial-time decidable, whereas Observation Selection is NP-complete even with two singleton explanations available across the entire table.
\end{theorem}

\begin{proof}[Proof sketch]
For listed contexts, guess a hypothesis set per cell and verify all its contexts; additionally guessing $O$ preserves NP. With $X=\varnothing$ and one context, Lemma 3 reduces SAT to either task. Tabulated families permit polynomial cell-intersection checks. Selection remains NP-hard by Hitting Set \cite{miller_reducibility_1972}, using Theorem 2's two-explanation construction.
\end{proof}

Contexts, explanations, and contracts can each require exponential space; the tabulated bound excludes their construction. A small obstruction rank bounds how many contexts a witness needs, but finding those contexts and their shared explanations can remain costly under a succinct representation. The contract is an exact semantic compilation rather than an assumed polynomial-size input to the succinct algorithms.

\begin{proposition}[Many necessary policy outputs]
There are polynomial-size acyclic Horn instances with $n$ source variables for which full observation is sufficient and every valid policy uses $2^n$ distinct explanations.
\end{proposition}

\begin{proof}
Realize $\bigwedge_i(y_i=x_i)$ by Lemma 3 over all contexts. Each admits only its matching complete assignment, forcing $2^n$ outputs from a linear-size theory.
\end{proof}

Distinct outputs need not require a large policy circuit: the matching policy is linear-size. The output bound instead explains why a full policy table is unsuitable as a polynomial certificate in the succinct setting. Validating a supplied policy circuit is coNP-complete: universally evaluate its hypothesis output and Horn closure at each possible context. The no-hypothesis tautology reduction gives hardness even for a constant policy.

Table~\ref{tab:complexity} summarizes the representation boundaries.

\begin{table*}[t]
\centering\small
\begin{tabular}{p{1.55in}p{2.20in}p{2.30in}}
\toprule
Representation & Sufficiency of a supplied scheme & Selection within a cardinality budget \\
\midrule
Circuit-defined contexts, Horn-defined explanations & $\Pi_3^P$-complete & $\Sigma_4^P$-complete \\
Listed contexts, Horn-defined explanations & NP-complete & NP-complete \\
Listed contexts and explanation families & P & NP-complete \\
\bottomrule
\end{tabular}
\caption{Complexity of sufficiency verification and observation selection across representations.}
\label{tab:complexity}
\end{table*}

\section{Structural Fragments, Composition, and Scope}

\subsection{When pairwise compatibility becomes sufficient}

The obstruction rank is $0$ if no obstruction exists, and otherwise $\max_{D\in\mathcal B(\mathcal P)}|D|$. With nonempty individual families, rank at most two means every pairwise-intersecting collection has a common explanation.

\begin{definition}[Rectangular explanation families]
Families are rectangular when, for some $t\ge1$, there are nonempty finite domains $L_1,\ldots,L_t$ and a fixed injective encoding $\iota:L_1\times\cdots\times L_t\rightarrow\mathcal U$ such that, for every $c$,
\end{definition}

\begin{equation}
\mathcal E_c=\{\iota(\ell_1,\ldots,\ell_t):
\ell_i\in A_i(c)\text{ for all }i\},
\qquad A_i(c)\subseteq L_i.
\label{eq:rectangular-families}
\end{equation}

\begin{lemma}[Order convexity and one-hot realization]
Each $\mathcal E_c$ is order-convex: if $S_1\subseteq S_2\subseteq S_3$ and $S_1,S_3\in\mathcal E_c$, then $S_2\in\mathcal E_c$. Every finite table of slot permissions $A_i(c)$ is Horn-realizable under the one-hot encoding $\iota(\ell_1,\ldots,\ell_t)=\{h_{i,\ell_i}:1\le i\le t\}$, with one distinct hypothesis for each slot-value pair.
\end{lemma}

\begin{proof}[Proof sketch]
Monotonicity preserves entailment from $S_1$ to $S_2$ and consistency from $S_3$ to $S_2$. For realization, prohibit two values per slot. Exact guards $v_c\leftarrow F_c$ enable $w_i\leftarrow v_c,h_{i,\ell}$ precisely for $\ell\in A_i(c)$; require all $w_i$ for $g$. Thus explanations select exactly one permitted value per slot. The construction is acyclic, polynomial in the table, and permits binary conjunction chains.
\end{proof}

Rectangularity concerns actual families: the abstract family $\{\varnothing,\{a,b\}\}$ cannot be Horn-realized because order convexity also requires $\{a\},\{b\}$. Acyclicity alone does not imply rectangularity.

\begin{theorem}[Tight slot bound]
For rectangular explanation families, every minimal obstruction has size at most $\max_i|L_i|$. This bound is tight. In particular, for Boolean slots and nonempty individual families, sufficiency is equivalent to separating every pair of contexts with disjoint explanation families.
\end{theorem}

\begin{proof}[Proof sketch]
Injectivity makes a product intersection nonempty exactly when every slot intersection is nonempty. A minimal obstruction therefore empties some slot. Choosing one context that excludes each value witnesses emptiness with at most $|L_i|$ contexts; minimality forces these to be all of $D$. One slot with $r$ contexts, each excluding its own value, attains rank $r$ under Lemma 4, including the empty singleton family for $r=1$. For Boolean slots with nonempty families, only pair obstructions remain.

The witness argument uses a single failing slot even if several slots fail: a proper subset emptying that slot would already be an obstruction. For tightness, all additional slots may permit their full domains, so they neither introduce smaller witnesses nor change the shared-choice condition. Theorem 1 then translates the Boolean case into the stated pair-separation criterion.
\end{proof}

For tabulated families of fixed rank $r$, subsets of at most $r$ contexts suffice to compile all minimal obstructions in polynomial time. Minimum observation selection remains NP-hard even at rank two by Theorem 2.

The distinction is between finding the requirements and choosing variables that meet them. Even when every requirement comes from a pair of contexts, its span can contain many variables. The resulting edges can encode an arbitrary Hitting Set instance, so pairwise compatibility simplifies witness construction without removing the optimization problem.

Without integrity constraints, explanation families are upward closed. If each is nonempty, all contain $H$, so no observation is needed; otherwise no scheme suffices. Example 2 shows why imposing local minimality changes this conclusion.

\subsection{Independent components share observation requirements}

Let $\mathcal P_1,\ldots,\mathcal P_m$, $m\ge1$, share $X,C$ and have pairwise disjoint hypothesis and derived vocabularies. Each rule body uses only its component's atoms and shared sources. Their combination $\mathcal P_\otimes$ has $H_\otimes=\bigcup_i H_i$ and $T_\otimes=\bigcup_i T_i\cup\{g_\otimes\leftarrow g_1,\ldots,g_m\}$, with fresh derived $g_\otimes$, shared sink $\bot$, and no additional rules or constraints.

\begin{theorem}[Contract composition]
Under these independence assumptions, $O$ is sufficient for $\mathcal P_\otimes$ exactly when it is sufficient for every $\mathcal P_i$. Moreover,
\end{theorem}

\begin{equation}
\mathcal H(\mathcal P_\otimes)
=\min_{\subseteq}\bigcup_i\mathcal H(\mathcal P_i).
\label{eq:contract-composition}
\end{equation}

\begin{proof}[Proof sketch]
Validity factors into component restrictions of the joint hypothesis set. On each cell, common component explanations combine exactly when each exists, since vocabularies are disjoint and constraints internal. Sufficiency is therefore conjunction of component conditions. Theorem 1 and Proposition 2 yield the irredundant union in Equation~\eqref{eq:contract-composition}.

In the forward direction, restricting a common joint explanation to $H_i$ supplies the component's common explanation. In the reverse direction, choose one common explanation per component and take their union. Internal rules preserve each component's derivation, and internal constraints cannot become violated through hypotheses from another component. Thus the same cellwise choices establish both directions of the equivalence.
\end{proof}

Extend the triage example by a freely varying source $r$. An independent component with permitted singleton labels $\{1,2\},\{1,3\},\{2,4\},\{3,4\}$ at $(q,r)=00,01,10,11$ has contract $\{\{q,r\}\}$, from its two incompatible diagonals. The joint contract has edges $\{p,q\},\{q,r\}$ and minimal schemes $\{q\},\{p,r\}$: sharing telemetry can lower observation cost. Cross-component constraints can invalidate Equation~\eqref{eq:contract-composition} by making separately valid explanations jointly inconsistent.

\subsection{Interpretation and limits}

The framework supports telemetry selection, disclosure verification, and modular observation design. Restricting available variables, assigning nonnegative costs, and conjoining independent requirements instantiate the same contract. Its guarantee covers every context in the supplied space, under simultaneous exact Boolean observations. Additional background information can remove obstructions through context restriction; incorrectly excluding a possible context can invalidate the intended guarantee. Hypotheses are permissible explanations, without probabilistic or causal guarantees. Sequential sensing, noisy reports, richer entailment, and preferences applied before uniform choice require separate analyses.

\section{Conclusion}

Minimal empty-intersection families characterize when partial observations permit a uniform abductive commitment. Their spans form canonical observation contracts, linking higher-order compatibility to hypergraph transversals. Tight slot bounds and independent composition identify structured cases. The $\Pi_3^P$ and $\Sigma_4^P$ completeness results separate information dependencies from polynomial Horn inference, while explicit representations expose lower complexity boundaries. The paired-sensor reduction explains how choosing an information channel can enforce a global commitment despite report-dependent explanatory choices. Sequential observation and uniform preference semantics are natural extensions.

\bibliographystyle{named}
\bibliography{References}
\end{document}